\documentclass{article}
\usepackage{amsmath,amsthm,amsfonts}
\usepackage{graphicx}
\usepackage{epstopdf}
\usepackage{multirow}
\def\smallddots{\mathinner{\raise7pt\hbox{.}\raise4pt\hbox{.}\raise1pt\hbox{.}}}
\def\smallsdots{\mathinner{\raise1pt\hbox{.}\raise4pt\hbox{.}\raise7pt\hbox{.}}}

\DeclareMathOperator{\diag}{diag}

\DeclareMathOperator{\nrank}{nrank}

\numberwithin{equation}{section}
\numberwithin{table}{section}
\newtheorem{theorem}{Theorem}[section]

\newtheorem{corollary}{Corollary}[section]

\newtheorem{algorithm}{Algorithm}[section]

\newtheorem{definition}{Definition}[section]

\newtheorem{remark}{Remark}[section]

\begin{document}

\title{\bf Real Polynomial Root-finding
 by Means of Matrix and Polynomial 
Iterations\footnote{This work appeared in Theoretical Computer Science, 2017,
http://dx.doi.org/10.1016/j.tcs.2017.03.032. It 
has been supported by NSF Grants CCF 1116736 and CCF-1563942
and PSC CUNY Award 67699-00 45. 
Some of its results have been presented at CASC 2014.}}
\author{Victor Y. Pan  $^{[1, 2],[a]}$ and Liang Zhao$^{[2],[b]}$  \\
\and\\
$^{[1]}$ Department of Mathematics and Computer Science \\
Lehman College of the City University of New York \\
Bronx, NY 10468 USA \\
$^{[2]}$ Ph.D. Programs in Mathematics  and Computer Science \\
The Graduate Center of the City University of New York \\
New York, NY 10036 USA \\
$^{[a]}$ victor.pan@lehman.cuny.edu \\
http://comet.lehman.cuny.edu/vpan/  \\
$^{[b]}$ lzhao1@gc.cuny.edu \\
} 
 \date{}

\maketitle


\begin{abstract}
Univariate polynomial root-finding is a classical subject, still 
important for modern computing. Frequently one seeks
just the real roots of a polynomial with real coefficients.
They can be approximated at a low computational cost if the
polynomial has no nonreal roots, but typically nonreal roots are 
much more numerous than the real ones. The subject of 
devising efficient real root-finders has been long and
 intensively studied. Nevertheless, we propose some novel 
ideas and techniques and obtain dramatic acceleration of
the known numerical algorithms. In order to achieve our progress 
we exploit the correlation
between the computations with matrices and polynomials, 
randomized matrix computations, and complex plane geometry,
 extend the techniques of the matrix sign iterations, and
 use the structure of the companion matrix of
the input polynomial.
The results of our extensive numerical tests with benchmark polynomials 
and random matrices are quite encouraging. In particular in these tests 
we have consistently computed accurate approximations of
 the real roots of benchmark polynomials 
of degree up to 1024 by using the IEEE standard double precision.
Moreover the number of iterations 
required for convergence of our algorithms
grew very slowly (if at all) as we increased 
the degree of the univariate input polynomials and the
dimension of the input  matrices from 64 to 1024.
  \end{abstract}



\section{Introduction}
Assume
a univariate polynomial of degree $n$  with
 real coefficients,
\begin{equation}\label{eqpoly}
 p(x)=\sum^{n}_{i=0}p_ix^i=p_n\prod^n_{j=1}(x-x_j),~~~ p_n\ne 0,
\end{equation}
which  has $r$ real roots 
  $x_1,\dots,x_r$ and $s=(n-r)/2$
pairs of nonreal complex conjugate roots.
In some applications, e.g., to algebraic and geometric optimization,
one seeks only the $r$ real roots, which typically
make up just a small fraction of all roots.\footnote{Recall the following excerpt from \cite{DGS14}: ``A celebrated result due to Erd\"{o}s and Tur\'an \cite{ET50} says that, for a univariate polynomial over C whose middle coefficients are not too big with respect to its extremal coefficients, the arguments of its roots are approximately equidistributed. Combined with a recent result of 
Hughes and Nikeghbali \cite{HN08}, this shows that the roots 
of such a polynomial clustered near the unit circle." 
This result does not apply to some classes of polynomials of practical 
importance, but the study in \cite{EGT10}
(also see 
the earlier papers \cite{K43} and \cite{EK95})
shows that
the expected number $r$ of the real roots of random real polynomials 
of various such classes stays in the range between orders of $\log (n)$ and $\sqrt n$; .} 
 The design of efficient real root-finders 
is a well studied subject (see
 \cite[Section 10.3.5]{EPT14}, \cite{PT13}, \cite{SMa},
and the  bibliography therein), but the 
most popular packages of subroutines 
for root-finding such as MPSolve 2000 
\cite{BF00}, Eigensolve 2001 \cite{F02},
and MPSolve 2012 \cite{BR14} 
approximate the $r$ real roots about as fast 
and as slow as all the $n$ complex roots.
It can be surprising, but we present some
novel methods that
 accelerate the known numerical real root-finders by a factor of $n/r$,
which is  dramatic in various important applications.
 
The springboard for our real root-finders is the matrix sign iterations,
which we apply to the companion matrix of 
an input polynomial.
It is a well known technique for matrix computations \cite{H08},
and we make it particularly efficient 
for real  polynomial root-finding,
although it has never been used for this purpose so far.
By combining it with a number of known and novel techniques 
we ensure fast convergence of  the iterations and 
their efficiency in  numerical implementation with the IEEE
 standard double precision. 

Our numerical tests 
confirm the efficiency of this approach. In 
particular, we closely approximate the real roots of 
 various benchmark polynomials 
of degree up to 1024 by using double precision.
Moreover the number of 
iterations required for convergence 
was typically quite small and  grew very slowly (if it grew at all)
as the polynomial degree increased from 64 to 1024.  

Some of our techniques should be of independent interest,
e.g., our numerical  
stabilization in Section \ref{snstbz}, our 
 exploitation of 
matrix functions and
randomized matrix computations in Algorithm \ref{alg2},
 and the combination of our maps of the complex plane with 
rational transformations of the variable
 and 
the roots.
Some of our algorithms (e.g., the ones of Section \ref{smbms2}) 
 combine
 operations with matrices and polynomials,
demonstrating once again the
value of 
  synergistic combinations of this kind, 
which we have been advocating since
\cite{P92} and \cite{BP94}.

Our goal in this paper is to present a novel approach to real root-finding 
for a  polynomial and to demonstrate 
the promise of this approach by performing some preliminary tests. 
We hope that we have advanced toward our goals substantially, 
and  there are promising directions for substantial improvement 
of  the implementation of our algorithms. 
  For example, Stage 3 of our Algorithm 3.1 is reduced to the inversion or orthogonalizaton of Toeplitz-like matrices, and the customary numerical algorithms, currently available for these operations, can be dramatically accelerated 
by means of the techniques of the papers  
 \cite{XXCB} and  \cite{XXG12}.

We organize our paper as follows.
In the next  section we cover some background material.
We present a variety  of our 
real polynomial root-finders in Section \ref{s4}. 
In Section \ref{snmt}
(the contribution of the second author)
 we present the results of our numerical tests. 
In the Appendix we cover some auxiliary results.
  

\section{Basic Definitions and Results}\label{s3}


Hereafter ``flop" stands for ``floating point arithmetic operation",
assumed to be performed numerically, with bounded  precision, e.g., 
the standard IEEE double precision.


\subsection{Some Basic Definitions for Matrix Computations}\label{sfnd}


\begin{itemize}  
  
\item
$\mathbb C^{m\times n}$ denotes the 
linear space of complex $m\times n$ matrices.
$\mathbb R^{m\times n}$ is its subspace of
 $m\times n$ real matrices. 

\item
$M^T=(m_{ji})_{i,j=1}^{m,n}$ is the transpose of a matrix $M=(m_{ij})_{i,j=1}^{m,n}\in \mathbb C^{m\times n}$.
$M^H$ is its Hermitian  transpose. $M^H=M^T$ for a real matrix $M$.

\item
$||M||=||M||_2$ denotes its spectral norm.

\item
$I=I_n$
is the $n\times n$ identity matrix.


\item
$\diag(b_j)_{j=1}^s=\diag(b_1,\dots,b_s)$ is the $s\times s$  diagonal matrix
with the diagonal entries $b_1$, $\dots$, $b_s$.

\item
$\mathcal R(M)$ is the range of
a matrix $M$,
that is, the linear space
generated by its columns.

\item
A matrix of full column rank is a {\em matrix basis} of its range.

\item
A matrix $Q$ is 
{\em unitary } if $Q^HQ=I$
 or $QQ^H=I$, and such a matrix is called  {\em orthogonal}
if it is real. 

\item
Suppose  
an $m\times n$ matrix $M$ has rank $n$
(and so
$m\ge n$). Write  $(Q,R)=(Q(M),R(M))$ to
denote a unique pair of a unitary $m\times n$ matrix $Q$ and
an upper triangular $n\times n$ matrix $R$ such that $M=QR$
and all diagonal entries of the matrix $R$
are positive  \cite[Theorem 5.2.3]{GL13}.

\item
$M^+$ is the unique Moore--Penrose pseudo inverse of $M$ \cite[Section 5.5.2]{GL13},
equal to $M^H$ 
if and only if the matrix $M$ is  unitary.

\item
An $m\times n$  matrix $M$ has an $n\times m$ 
{\em left inverse} matrix $X=M^{(I)}$ such that  $XM=I_n$
 if and only  if it has full column rank $n$. In this case 
$M^+$ is a left inverse. 
 The left inverse is
 unique if and only if  $M$ is
 a nonsingular matrix, in which case $m=n$
 and $M^{(I)}=M^{-1}$. 

\item
The $\epsilon$-rank of a matrix $M$ is
the minimal rank of the matrices in 
its $\epsilon$-neighborhood. 
{\em Numerical  rank} $\nrank (M)$ is
the $\epsilon$-rank where $\epsilon$
is small in context.
\end{itemize}

\begin{definition}\label{defeig} {\rm Eigenvalues,  eigenvectors and eigenspaces.}
\begin{itemize}  
  
\item
A scalar  $x$ is an
{\em eigenvalue}  of a matrix $M$
 associated with an {\em eigenvector}  ${\bf v}$ if
$M{\bf v}=x{\bf v}$.
\item
The eigenvectors associated with an
eigenvalue $x$ or, more generally, with any set of the eigenvalues $\mathcal X\in\mathcal X(M)$
 form the {\em eigenspaces}
$\mathcal S(M,x)$ and $\mathcal S(M,\mathcal X)$, respectively,
 associated with 
 the eigenvalue $x$ and the set $\mathcal X$ of eigenvalues, respectively.
A linear subspace $\mathcal S$ of $\mathbb C^{n\times n}$
is an  eigenspace of a matrix $M$ 
if and only if
$M\mathcal S=\{M{\bf v}:{\bf v}\in \mathcal S\}\subseteq \mathcal S$
(see \cite[Definition 4.1.1]{S01}).
\item
An eigenvalue $x$ of a matrix $M$ is a root of the characteristic 
polynomial $\det(xI-M)$. The multiplicity of this root
is the {\em algebraic multiplicity} of the eigenvalue $x$,
denoted  $am(x)$. 
The dimension $gm(x)=\dim(\mathcal S(M,x))$
 is the {\em geometric multiplicity} of $x$,
never exceeding $am(x)$. 
An eigenvalue $x$ is {\em simple} if $gm(x)=1$. 
\end{itemize}
\end{definition}


\subsection{The Companion Matrix and the Frobenius Algebra}\label{scmpn}

Let ${\bf e}_n^T=(0,0,\dots,0,1)$ denote the $n$th coordinate vector and write
${\bf p}=(p_i/p_n)_{i=0}^{n-1}$,
 \begin{equation}\label{eqzcrc}
Z=C_0=\begin{pmatrix}
        0   &       &       &   & 0 \\
        1   & \ddots    &       &   & 0 \\
            & \ddots    & \ddots    &   & \vdots    \\
            &       & \ddots    & 0 & 0 \\
            &       &       & 1 & 0 \\
    \end{pmatrix}~{\rm and}~C_p=\begin{pmatrix}
        0   &       &       &   & -p_0/p_n \\
        1   & \ddots    &       &   & -p_1/p_n \\
            & \ddots    & \ddots    &   & \vdots    \\
            &       & \ddots    & 0 & -p_{n-2}/p_n \\
            &       &       & 1 & -p_{n-1}/p_n \\
    \end{pmatrix}=Z-{\bf p}^T{\bf e}_n. 
\end{equation}

$Z$ is the down-shift matrix.
$C_p$ is the {\em companion matrix} of the polynomial $p(x)$ of
(\ref{eqpoly}), which is the characteristic polynomial of this matrix.
Hence real root-finding for the polynomial $p(x)$ turns into real
eigen-solving for this matrix.
 
 $Z{\bf v}=(v_{i-1})_{i=1}^{n}$,
for a vector ${\bf v}=(v_i)_{i=1}^{n}$ and for $v_0=0$.



\begin{theorem}\label{facfr} ({\em The Cost of Computations in the Frobenius Matrix Algebra}.)
The companion matrix $C_p\in \mathbb R^{n\times n}$ of 
a polynomial $p(x)$ of (\ref{eqpoly})
generates the Frobenius matrix algebra $\mathcal A_p$.
 One needs  $O(n)$ flops for  addition and
$O(n\log (n))$  for multiplication and
inversion  in  this algebra as well as for
multiplication of a matrix in this algebra
  by a vector. These cost bounds hold both for 
exact computation with no errors and  
numerically stable approximate computations.
\end{theorem}

\begin{proof}
The estimates for the exact computation
can be readily deduced from the following expressions from \cite{Pa}
for a matrix $C({\bf p,c})$ in the algebra  $\mathcal A_p$ 
defined by  its first column ${\bf c}=(c_i)_{i=0}^{n-1}$,
\begin{equation}\label{eqfrobm}
C({\bf p,c})= \sum_{i=0}^{n-1}c_iC_p^i=\sum_{i=0}^{n-1}c_iZ^i+{\bf p}{\bf u}^T,
\end{equation}
for ${\bf p}$ of equation (\ref{eqzcrc}),
${\bf u}=(u_i)_{i=0}^{n-1}$, $u_0=0$,
$u_{i}=\sum_{k=n-i}^{n-1}c_k\rho^{k-n+i}$, for $i=1,\dots,n-1$
 and $\rho=-p_{n-1}/p_n$, so that $\rho=0$ if $p_{n-1}=0$. 
(If $c_0=0$, then  invert the matrix $C({\bf q,c_s})=C({\bf p,c})$
 for $q(x)=p(x-s)$ and a random 
real or complex shift $s$.) 
The algorithm of \cite{C96}, using the transition to the so called
 ``Horner's basis", performs numerically stable multiplication in 
the algebra $\mathcal A_p$, 
\cite{XXCB}  performs 
numerically stable inversion by using  $O(n\log^2 (n))$ flops,
which is accelerated by a factor of $\log (n)$ in  \cite[Section 9.8]{P15}.
\end{proof}


\subsection{Decreasing the Size of an Eigenproblem}\label{seig}


An eigenvalue $x$  of a matrix $M$ as well as a set of eigenvalues $\mathcal X$ 
are {\em dominant} 
if they are absolutely larger than all the other eigenvalues. 
An eigenspace is called {\em dominant} 
if it is associated  with a dominant eigenvalue
or a dominant  set of eigenvalues.

The set
$\mathcal X(M)$ of all eigenvalues of a matrix $M$ 
is called its {\em spectrum}.

The Power Method \cite[Section 7.3.1]{GL13}
computes the vector $M^k{\bf v}$, for a random vector ${\bf v}$
and a sufficiently large integer $k$. 
The 1-dimensional vector space $\{tM^k{\bf v}\}$, for $t\in \mathbb C$,
is  expected to
approximate the eigenspace associated with an eigenvalue
$x$ 
if it is  dominant and simple.
This would not work
only if  
the vector ${\bf v}$ has an absolutely small component
along the eigenvector 
associated with this eigenvalue $x$,
but such an event is unlikely, for a
random vector  ${\bf v}$. 
One can choose $k=1$ if the domination of the eigenvalue
$x$ in the spectrum of $M$
is strong. 
Let us extend the Power Method  for $k=1$ 
to the
approximation of a strongly dominant eigenspace of a dimension $r$. 

\begin{algorithm}\label{alg0} {\rm Approximation of the dominant eigenspace.}

    \item{\textsc{Input:}} an $n\times n$ matrix $M$, the dimension $r$ 
of its dominant  eigenspace  $\mathcal U$, $0<r<n$, and two tolerance
bounds: a positive integer $K$ and a positive $\epsilon$.

    \item{\textsc{Output:}} FAILURE (with a low probability) 
    or a unitary matrix $U$ whose range approximates
    the eigenspace $\mathcal U$. 

    \item{\textsc{Computations:}}
         
\begin{enumerate}
  \item
  Apply the randomized algorithm of \cite{HMT11},
which at first generates a  standard Gaussian random $n\times r_+$ matrix $G$
for a proper integer $r_+>r$ and
         then computes the matrix $H=MG$ and 
 the numerical rank $\nrank(H)$. 
 \item
Unless  $\nrank(H)=r$, re-apply the algorithm of \cite{HMT11} up to $K$ times
until the equation $\nrank(H)=r$ is  observed.
If it is never observed, output FAILURE (this occurs with a probability near 0).
    \item
  If $\nrank(H)=r$, then compute the QR factorization
 $H=Q(H)R(H)$, 
output an $n\times r$ unitary matrix $U$ approximating the first $r$ columns of
the matrix $Q(H)$, and  stop. (The analysis in
 \cite[Section 4]{HMT11}, \cite[Section 7.4]{PQYa}, and \cite[Theorem 4.3]{PZa} shows that,
 with a probability close to 1, the columns of the matrix $U$
closely approximate a unitary basis of the eigenspace 
$\mathcal U$ and that $||M-UU^HM||\le \epsilon ||M||$. The latter bound would 
certify correctness of the output.)
\end{enumerate}
 \end{algorithm}

The arithmetic cost of performing the algorithm is $O(n^2r_+)$,
but decreases to $O(nr_+(r_++\log(n))$, for $M=C_P$,
by virtue of Theorem \ref{facfr}.
It increases by a factor of $\log (r)$ if the
dimension $r$ of the eigenspace $\mathcal U$
is not available, but is computed by using binary search 
that begins with recursive doubling of
the candidate integer values 1, 2, 4, etc. 
The algorithm generates $nr_+$ random parameters,
but its modification using the structured (so called  SRFT) multipliers $G$
involves only $n$ such parameters and only $O(n\log(n))$ flops for
the computation of the product $MG$
(see  \cite[Section 11]{HMT11} and \cite[Section 7.5]{PQYa}).
Alternative application of the orthogonal iterations 
of \cite[Section 7.3.2]{GL13} requires order of $n^2r_+$ flops.

\begin{remark}\label{rerrqr}
Actually the algorithm of \cite{HMT11} works even 
where the input includes
 an 
upper bound $r_+$ on the dimension $r$ 
of the dominant  eigenspace  $\mathcal U$,
rather than the dimension itself,
and then the algorithm can
 compute this dimension $r$  within the above computational cost
as by-product. (The integer $r=\nrank(H)$ can be
obtained, e.g., from
rank revealing QR factorization of the matrix $H$.)
\end{remark}

Now suppose that we have an eigenspace generated by $r$ 
eigenvalues of an $n\times n$ matrix. Then the following simple theorem 
(extending the recipe of the Rayleigh quotients) enables us to
approximate these eigenvalues    
as the eigenvalues of 
an auxiliary $r\times r$ matrix. 


\begin{theorem}\label{thsubs} ({\em Decreasing the Eigenproblem Size to the Dimension of
an Eigenspace}, cf.
 \cite[Section 2.1]{W07}.)

Suppose that $M\in\mathbb C^{n\times n}$,
 $U\in \mathbb C^{n\times r}$, and the matrix $U$
has full column rank $r\le n$ and generates 
 the space $\mathcal U=\mathcal R(U)$.
 Then 

(i) $\mathcal U$ is an eigenspace of $M$
if and only if there exists a matrix $L\in\mathbb C^{k\times k}$
such that $MU=UL$ or equivalently if and only if $L=U^{(I)}MU$,

(ii) $\mathcal X(L)\subseteq \mathcal X(M)$,

(iii) $MU{\bf v}=x U{\bf v}$
if $L{\bf v}=x {\bf v}$,

(iv) the matrix $L$ is unique, that is, its choice is independent
of the choice of a matrix $U$ and its left inverse $U^{(I)}$, and so
$L=U^HMU$, 
for a unitary matrix $U$.
\end{theorem}


The algorithm and the theorem enable us to
approximate the $r$ real eigenvalues of a matrix
as the 
$r$ dominant eigenvalues of an auxiliary matrix. 
Theorems  \ref{thsubs} and \ref{thsmf}  (below) together  
suggest a direction to such a reduction, and
we achieve it in Sections \ref{smbms1} and \ref{snstbz}.


\subsection{Matrix Functions and Eigenspaces}\label{seig1}


\begin{theorem}\label{thsmf} ({\em The Eigenproblems for a Matrix
and Its Function}.) 

Suppose that $M$ is a square matrix and that a rational
function $f(x)$ is defined  on its spectrum. 

 (i) Then $f(M){\bf v}=f(x){\bf v}$ if $M{\bf
v}=x{\bf v}$. 

(ii) Let $\mathcal U=\mathcal U_{\mu,f}$ denote 
the eigenspace of the matrix $f(M)$ associated
with its eigenvalue $\mu$. Then this is an eigenspace of the
matrix $M$ associated with all its
eigenvalues $x$ such that $f(x)=\mu$. 

(iii) The space
$\mathcal U$ has dimension 1 and
is associated with a single eigenvalue of $M$ if
$\mu$ is a simple eigenvalue of $f(M)$.
\end{theorem}


\begin{proof}
We readily verify part (i), which implies parts (ii) and (iii).
\end{proof}


\begin{remark}\label{redim}
 The matrix $Z^k$,
for $1\le k\le n$, has
the single eigenvalue 0  satisfying $am(0)=n$ and $gm(0)=k$, and so
  $\dim(\mathcal U_{0,f})=k$,
for  $M=Z$, $f(x)=x^k$, and $k=1,\dots,n$.
\end{remark}


Suppose that we have computed a matrix basis
 $U\in \mathbb C^{n\times r_+}$,  for
an  eigenspace  $\mathcal U$ of  a matrix function $f(M)$ of
an $n\times n$ matrix $M$. By virtue of Theorem \ref{thsmf},
this is  a matrix basis of an eigenspace of the matrix $M$.
We can first compute a left inverse $U^{(I)}$
or the
orthogonalization $Q=Q(U)$
and then
approximate the eigenvalues of $M$
associated with this eigenspace as the eigenvalues of the
$r_+\times r_+$ matrix $L=U^{(I)}MU=Q^HMQ$
(cf. Theorem \ref{thsubs}). 

If $r=1$, then the matrix $U$ turns into an eigenvector ${\bf u}$,
shared by the matrices $f(M)$ and $M$,
while the matrix $L$ turns into the  
 the Rayleigh Quotient
$\frac{{\bf u}^TM{\bf u}}{{\bf u}^T{\bf u}}$
or the simple quotient $(M{\bf u})_i/u_i$, for any $i$ 
such that $u_i\neq 0$.




\subsection{Some Maps in the Frobenius Matrix Algebra}\label{smbms3}


Part (i) of
Theorem \ref{thsmf} implies that, 
for a polynomial $p(x)$ of (\ref{eqpoly}) and a rational function
$f(x)$ defined on the set $\{x_i\}_{i=1}^n$ of its roots, the
rational matrix function $f(C_p)$ has the spectrum
$\mathcal X(f(C_p))=\{f(x_i)\}_{i=1}^n$. In particular, the maps
$$C_p\rightarrow C_p^{-1},~C_p\rightarrow aC_p+bI,~
C_p\rightarrow C_p^2,~C_p\rightarrow \frac{C_p+C_p^{-1}}{2},~{\rm and}~
C_p\rightarrow \frac{C_p-C_p^{-1}}{2}$$
induce the maps of the eigenvalues of the matrix $C_p$, and thus induce
 the maps of the roots of its characteristic polynomial $p(x)$
given by the equations
$$y=1/x,~y=ax+b,~y=x^2,~y=0.5(x+1/x),~{\rm and}~y=0.5(x-1/x),$$
respectively.
The latter two maps can be only applied if the matrix $C_p$ is nonsingular,
so that $x\neq 0$, and similarly for the two dual maps below.

By using the reduction modulo $p(x)$,
define the five dual maps
\begin{eqnarray*}
y=(1/x) \mod p(x),~y=ax+b\mod p(x),~
y=x^2\mod p(x),\\
~y=0.5(x+1/x)\mod p(x),~
{\rm and}~y=0.5(x-1/x)\mod p(x),
\end{eqnarray*}
where $y=y(x)$ denotes polynomials.
Apply the two latter maps recursively, to
define two iterations with polynomials modulo $p(x)$ as follows,
$y_0=x$, $y_{h+1}=0.5(y_h+1/y_h)\mod p(x)$ and
$y_0=x,~y_{h+1}=0.5(y_h-1/y_h)\mod p(x)$,
$h=0,1,\dots$.
More generally, define the iteration
$y_0=x$, $y_{h+1}=ay_h+b/y_h\mod p(x)$, $h=0,1,\dots$,
 for any pair of scalars $a$ and $b$,
provided that $y_h=0$, for none $h$. 


\section{Real Root-finders with Modified Matrix Sign Iterations. Variations and Extensions}\label{s4}


In this section we present some efficient numerical real root-finders
based on modification of the matrix sign classical iterations applied to the companion matrix of 
the input polynomial.

\subsection{A Modified Matrix Sign Iterations}\label{smbms1}


Our first algorithm approximates the $r$ real roots of a polynomial $p(x)$
of (\ref{eqpoly}) 
as the real eigenvalues of the companion matrix $C_p$.
It applies the matrix iterations
\begin{equation}\label{eqmsn}
M_0=C_p,~M_{h+1}=0.5(M_h-M_h^{-1}),~{\rm for}~ h=0,1,\dots,
\end{equation}
which modify the {\em matrix sign} iterations
$\widehat M_{h+1}=0.5(\widehat M_h+\widehat M_h^{-1})$ 
(cf.   \cite{H08}). 

For every eigenvalue $x_j$ of the matrix $M_0=C_p$,
 define 
its {\em  trajectory} made up of the eigenvalues of the matrices $M_h$,
being  its images in the maps $M_0\rightarrow M_h$, for $h=1,2,3,\dots$.
More generally 
iterations (\ref{eq20l}) below
 modifying the  M{\"o}bius classical iterations 
$x^{(h+1)}=\frac{1}{2}(x^{(h)}+1/x^{(h)}),~{\rm for}~h=0,1,\dots$,
define a trajectory initiated at any complex point $x^{(0)}$.

Hereafter we write sign$(z)=$sign$(\Re (z))$, for a complex number $z$
(cf.  \cite[page 107]{H08}).

\begin{theorem}\label{thsqnit} ({\em Convergence of the modified M{\"o}bius Iterations}.)
Fix a complex $x=x^{(0)}$ and
   define the modified  M{\"o}bius iterations

\begin{equation}\label{eq20l}
x^{(h+1)}=\frac{1}{2}(x^{(h)}-1/x^{(h)}),~{\rm for}~h=0,1,\dots.
\end{equation}

(i) The values $x^{(h)}$ are real, for all $h$,
if $x^{(0)}$ is real.  

(ii)  $|x^{(h)}-{\rm sign}(x)\sqrt {-1}|\le \frac{2K^{2^{h}}}{1-K^{2^h}}$,
 for
$K=|\frac{x-{\rm sign}(x)}{x+{\rm sign}(x)}|$
and
$h=0,1,\dots$.
\end{theorem}
\begin{proof} 
Part (i) is immediately verified.
Part (ii) readily extends
the similar estimate on \cite[page 500]{BP96}.
\end{proof}

 Theorem \ref{thsqnit} implies the following result.
 
\begin{corollary}\label{cotj}
As $h\rightarrow \infty$, the trajectories of the
$2s$ nonreal eigenvalues of $M_0=C_p$
 converge to $\pm \sqrt {-1}$ 
with the quadratic rate of convergence right from the start,
whereas the  trajectories of the
$r$ real eigenvalues are real, for all $h$.
\end{corollary}

\begin{algorithm}\label{alg2} {\rm Matrix sign iterations modified for real eigen-solving.}

 \item{\textsc{Input:}} two integers $n$ and $r$, $0<r<n$, and
    the coefficients of a polynomial $p(x)$ of equation (\ref{eqpoly}).

    \item{\textsc{Output:}}
    approximations to the real roots $x_1,\dots,x_r$ of the polynomial $p(x)$
or FAILURE with a probability close to 0.

    \item{\textsc{Computations:}}
         \begin{enumerate}
         \item
         Write $M_0=C_p$ and recursively compute the matrices $M_{h+1}$ of
(\ref{eqmsn}),
 for $h=0,1,\dots$ (cf.  Corollary \ref{cotj}).

         \item
         Fix a 
	 sufficiently 
	 large integer $k$ and compute the matrix $M=M_k^2+I_n$.

         (By extending Corollary \ref{cotj} observe that the map $M_0=C_p\rightarrow M$ sends all nonreal eigenvalues of $C_p$
         into a small neighborhood of the origin 0 and sends all real
         eigenvalues of $C_p$ into the ray $\{x:~x\ge 1\}$.)

         \item
          Apply our randomized Algorithm \ref{alg0} in order
         to approximate a unitary matrix $U$ whose columns
form a
basis for the eigenspace
associated with the $r$ dominant eigenvalues of the matrix $M$.

(By virtue of Theorem \ref{thsmf}, this is expected to be the eigenspace associated with
the 
real eigenvalues of the
matrix $C_p$, although
 with a probability close to 0
the algorithm can
 output FAILURE, in which case we stop the computations.)

     \item
Compute and output
approximations to the $r$ eigenvalues of the $r\times r$ matrix  $L=U^HC_pU$.
(They approximate the $r$ real eigenvalues of the matrix $C_p$ by virtue of
 Theorem \ref{thsubs} and consequently  approximate the $r$ real roots
of the polynomial $p(x)$.)
    \end{enumerate}
    \end{algorithm}


Stages 1 and 2 involve
$O(kn\log (n))$ flops by virtue of Theorem \ref{facfr}.
Stage 3  adds $O(n^2r)$ flops and  the cost $a_{rn}$ of
 generating $n\times r$ standard Gaussian random matrix.
Add  $O(r^3)$ flops performed at Stage 4 and 
arrive at the overall arithmetic cost bound
$O((kn\log (n)+nr^2)+a_{rn}$.


\begin{remark}\label{renrei} ({\rm Counting Real Eigenvalues.})
The binary search can produce  the number of real eigenvalues
as the numerical rank of the matrices $M_k^2+I$
when this rank stabilizes as $k$ increases. 
As the number of real roots increases, 
so does the size of the matrix $L$. This has consistently implied 
the decrease of the accuracy of the output approximations in our tests
(see the test results in Section \ref{snmt}).  
One can refine these approximations by applying the inverse 
Power Method or Newton's iterations, but 
if the accuracy becomes too low, one must extend the precision of computing. 
\end{remark}


\begin{remark}\label{refracclrb} ({\rm Acceleration by Means of Scaling.})
One can dramatically accelerate the initial convergence of
Algorithm \ref{alg2} by applying {\em determinantal scaling}  (cf.
\cite{H08}), that is, by replacing the matrix $M_0=C_{p}$ by the matrix
$M_{0}=0.5(\nu C_p-(\nu C_p)^{-1})$, for
$\nu=1/|\det(C_p)|^{1/n}=|p_n/p_0|$.
\end{remark}


\begin{remark}\label{renrei0} {\rm Real and Nearly Real Roots.}


In the presence of rounding errors Algorithm \ref{alg2} and all
our other algorithms approximate   both 
$r$ real and $r_+-r$ nearly real eigenvalues of the matrix $M$,
for some $r_+\ge r$.
The $r$ real eigenvalues, however, are the roots of $p(x)$
and
we can refine their approximations very fast
(cf. Theorems \ref{thren} and \cite{PT14}),
under some mild assumptions about the 
isolation of every such a root from the $n-1$ other roots.
(One can partly relax these assumptions by extending the techniques of \cite{PT14}.)
Then we can readily
select the $r$ real eigenvalues among the $r_+$ real and nearly real ones.

Generally, however, the distinction between 
real and nearly real roots is very slim in our numerical algorithms.
As this was pointed out in \cite{Ba}, 
in the course of performing the iterations, the real eigenvalues can
become nonreal, due to rounding errors, and then would
converge to $\pm \sqrt {-1}$. 
In our extensive tests we 
have never observed such a phenomenon, 
apparently because in these tests the  
convergence  to $\pm \sqrt {-1}$
was much slower 
for the nearly real eigenvalues
than  for the  eigenvalues with 
reasonably large imaginary parts. 
\end{remark}


\subsection{Inversion-free Variations of 
 the Modified Matrix Sign Iterations and Hybrid Algorithms}\label{sinvfr}



The overall arithmetic cost of the Modified Matrix Sign Iterations
is dominated by the cost of $k$ matrix inversions, that is, $O(kn\log^2(n))$
flops (cf. Theorem \ref{facfr}). 
 If all nonreal eigenvalues of the matrix $M_0$
lie in the two discs $D(\pm \sqrt{-1},1/2)=\{x:~|x\pm \sqrt{-1}|\le 1/2\}$, then 
we can avoid matrix inversions in 
the Modified Matrix Sign Iterations
by   
replacing iterations (\ref{eqmsn})
with any of the two iteration processes
\begin{equation}\label{eqpd3}
M_{h+1}=0.5(M_h^{3}+3 M_h)
\end{equation}
and
\begin{equation}\label{eqpd5}
M_{h+1}=-0.125(3 M_h^{5}+10 M_h^{3}+15 M_h),
\end{equation}
 for $h=0,1,\dots$.
Right from the start both
iterations
send the nonreal roots 
 toward the two points $\pm \sqrt{-1}$
 with quadratic and cubic convergence rates,
respectively. (In order to prove this, extend the proof of \cite[Proposition
4.1]{BP96}.) Both iteration processes keep the real roots real
 and use $O(n\log (n))$ flops per iteration.

What if the nonreal roots do not lie in these discs?
We can apply the following combination of iterations 
(\ref{eqmsn})--(\ref{eqpd5})
and Corollary \ref{cownmb} of Section \ref{srrd}. 

\begin{algorithm}\label{alghbr} {\rm A Hybrid Algorithm.}

    \item{\textsc{Input, Output}  as in Algorithm \ref{alg2}.}
\item{\textsc{Computations:}}
          Perform
 the iterations of
Algorithm \ref{alg2} until  
 a test shows that the $2s$ nonreal eigenvalues of the
input companion matrix are mapped into the discs $D(\pm \sqrt{-1},1/2)$. 
(For testing this condition, apply the algorithm that supports  
Corollary \ref{cownmb}. To keep the computational cost down,
apply this test periodically, according to a fixed policy,
based on heuristic rules or the statistics
of the previous tests.)
Then
 shift the computations
 to the inversion-free iterations 
(\ref{eqpd3}) or (\ref{eqpd5}) converging faster and using
 $O(n\log (n))$ flops per iteration. 
    \end{algorithm}

Let us recall some alternative matrix iterations for real root-finding without inversions.
Recall that sign$(M)=M(M^2)^{-0.5}$
and apply
 the Newton--Schultz iterations for 
the approximation of the matrix square root  \cite[equation (6.3)]{H08},

$$Y_{k+1}=0.5~Y_k(3I-Z_kY_k),~~Y_0=M^{-2},$$
and
$$Z_{k+1}=0.5~(3I-Z_kY_k)Z_k,~~Z_0=I,$$
for $k=0,1,\dots$.
The iterations keep real eigenvalues real and converge if $||I-M^{-2}||_p<1$,
for $p=1,2,$ or $\infty$. This assumption is 
easy to satisfy by means of scaling $M\rightarrow aM$, which keeps real eigenvalues real, 
for real $a$.

The similar coupling technique of \cite{Ba} is even simpler, because it is
 applied directly to the modified matrix sign iterations (\ref{eqmsn}),
preserving its quadratic convergence to  $\pm \sqrt{-1}$ right from the start.

In our tests
for numerical
real root-finding, however, we could perform safely only
a small number of
these inversion-free iterations at the initial stage,
 and then the images of the real eigenvalues of the matrix $C_p$ grew very large and
the condition numbers of the computed matrices blew up.


\subsection{Numerical Stabilization of 
 the Modified Matrix Sign Iterations}\label{snstbz}


The images of nonreal eigenvalues of the matrix $C_p$ converge to
$\pm \sqrt {-1}$ in the iterations of Stage 1
of Algorithm
  \ref{alg2}, but if the
images of some real eigenvalues of $C_p$ come close to 0,
then at the next step we would have to invert an ill-conditioned matrix $M_h$
unless we are applying an inversion-free variant of
the iterations of the previous subsection.

We can try to avoid this problem by
  shifting the matrix  (and its eigenvalues), that is,
by adding to the current matrix $M_h$
the matrix $s I$, for a reasonably 
small positive scalar $s$ or $-s$. We can select this scalar
 by applying heuristic methods or randomization.
In our tests this policy has 
preserved convergence quite well, but
we have no formal support for this observation.
The following stabilization of Algorithm   \ref{alg2}
 involves nonreal values 
even when the matrix $C_p$ was real,
but has both formal and empirical support.

\begin{algorithm}\label{alg2a} {\rm Numerical stabilization of the modified matrix sign iterations.}

    \item{\textsc{Input, Output} and Stages 3 and 4 of \textsc{Computations} are as in Algorithm \ref{alg2}, except that the input includes a small positive  scalar $\alpha$ 
such that no eigenvalues of the matrix $C_p$ have 
imaginary parts close to $\pm \alpha \sqrt {-1}$
(see Remark \ref{realph} below),
    the set of $r$ real roots $x_1,\dots,x_r$ of the polynomial $p(x)$
is replaced by the set of its $r_+$ roots having the imaginary parts in the
range
$[-\alpha,\alpha]$, and the integer $r$ is replaced by the integer $r_+$ throughout.}
  	
    \item{\textsc{Computations:}}
         \begin{enumerate}
         \item
         Apply Stage 1 of Algorithm \ref{alg2} to the two matrices 
         $M_{0,\pm}=\alpha\sqrt {-1}~I\pm C_p$, thus producing  two sequences 
	of the matrices $M_{h,+}$ and  $M_{h,-}$, for $h=0,1,\dots$. 

         \item
         Fix a 
	 sufficiently 
	 large integer $k$ and compute the matrix $M=M_{k,+}+M_{k,-}$.
 
    \end{enumerate}
    \end{algorithm}

Because of the assumed choice of $\alpha$, the matrices  $\alpha\sqrt {-1}~I\pm C_p$ have
no real eigenvalues, and so the images of all their eigenvalues, that is, the eigenvalues of the matrices
$M_{k,+}$ and $M_{k,-}$, converge to 
$\pm \sqrt {-1}$ as $k \rightarrow \infty$. Moreover, one can verify that the eigenvalues
of the matrix $M_{k,+}+M_{k,-}$ converge to 0 unless they are the images of the $r_+$ eigenvalues of the
matrix $C_p$ having the imaginary parts in the range $[-\alpha, \alpha]$. The latter   eigenvalues of the matrix  $M_{k,+}+M_{k,-}$
converge to $2\sqrt {-1}$. This shows correctness and numerical stability of 
Algorithm \ref{alg2a}.

The algorithm approximates the $r_+$ roots of $p(x)$ 
by using $O(kn\log(n)+nr_+^2)+a_{r_+n}$ flops,
versus $O(kn\log(n)+nr^2)+a_{rn}$ involved in
Algorithm \ref{alg2}.

\begin{remark}\label{realph}
We can test the proximity of the roots to a line in two stages:
by at first moving the line into the unit circle 
$\{x:~|x|=1\}$ 
(cf. Theorem \ref{thcl})
and then applying algorithms that supports 
Theorem  \ref{thrrd} or
Corollary \ref{cownmb}.
\end{remark}


\subsection{Square Root Iterations (a Modified Modular Version)}\label{smbms2}


Next we describe a dual polynomial version of Algorithm \ref{alg2}.
It extends the square root iterations 
$y_{h+1}=\frac{1}{2}(y_h+1/y_h)$, $h=0,1,\dots$,
and at Stage 2 involves the computation 
of the polynomial ${\rm agcd}(p,t_k)$, which denotes 
an {\em approximate
greatest common divisor} of the input polynomial
$p=p(x)$ and an auxiliary polynomial $t_k=t_k(x)$.
We refer the reader to
 \cite{P01b},
 \cite{KYZ05}, \cite{BB10},  \cite{WH10}, \cite{BR11},
and \cite{SSa}
for the definitions of this concept and the algorithms
for its computation.

Compared to Algorithm \ref{alg2}, we replace
all rational functions in the matrix $C_p$ by the same rational functions
in the variable $x$ and reduce them  modulo
the input polynomial $p(x)$.
The reduction does not affect the values of the functions at the roots
of $p(x)$, and it follows that these values are precisely the eigenvalues of
the rational matrix functions computed in Algorithm \ref{alg2}.

\begin{algorithm}\label{alg3} {\rm Square root modular iterations modified for real root-finding.}

    \item{\textsc{Input:}} two integers $n$ and $r$, $0<r<n$, and
    the coefficients of a polynomial $p(x)$ of equation (\ref{eqpoly}).

    \item{\textsc{Output:}}
    approximations to the real roots $x_1,\dots,x_r$ of the polynomial $p(x)$.


    \item{\textsc{Computations:}}
         \begin{enumerate}
         \item
(Cf. (\ref{eqmsn}).)         Write $y_0=x$ and  compute the polynomials
\begin{equation}\label{eqsqrt}
y_{h+1}=\frac{1}{2}(y_h-1/y_h)\mod p(x),~h=0,1,\dots.
\end{equation}

         \item
         Periodically, for some selected integers $k$,
compute the polynomials  $t_k=y_k^2+1\mod p(x)$.

     \item
    Write 
 $g_k(x)={\rm agcd}(p,t_k)$
and compute $d_k=\deg(g_k(x))$.
If $d_k=n-r=2s$,
compute the polynomial $v_k\approx p(x)/g_k(x)$
of degree $r$. Otherwise  continue the iterations of Stage 1.

     \item
     Apply one of the algorithms of  \cite{BT90}, \cite{BP98},
 and \cite{DJLZ97}
         (cf. Theorem \ref{thmlg}) to
     approximate the $r$ roots $y_1,\dots,y_r$ of the polynomial $v_k$.
 Output these approximations.
        \end{enumerate}
    \end{algorithm}

Our comments preceding this algorithm show that the values of the polynomials
$t_k(x)$ at the roots of $p(x)$ are equal to the images
of the eigenvalues of the matrix $C_p$ in Algorithm \ref{alg2}.
Hence the values of the polynomials $t_k(x)$ at the nonreal roots 
of $p(x)$ 
converge to 0
as $k\rightarrow \infty$, whereas their values at
 the real roots of $p(x)$ stay far from 0. Therefore, for sufficiently large integers $k$,
 ${\rm agcd}(p,t_k)$ turns into the polynomial $\prod_{j=r+1}^n(x-x_j)$.
This implies correctness of the algorithm.

Its asymptotic computational cost is $O(kn\log^2(n))$ plus the cost of computing
${\rm agcd}(p,t_k)$ and
choosing the integer $k$ (see  our next remark).

\begin{remark}\label{resmlr}
The latter algorithm
reduces real root-finding essentially
to the computation of agcd$(p,t_k)$.
One can apply quite efficient heuristic algorithms for this
computation (cf. \cite{P01b},  \cite{KYZ05}, \cite{BB10}, 
 \cite{WH10}, \cite{BR11}, and \cite{SSa}), but no good 
formal estimates  are available for
their complexity. One can, however,
note that $p(x)u_k(x)\approx t_k(x)v_k(x)$, and so,
assuming that $v_k(x)$ is a monic polynomial
(otherwise we can scale it), can obtain its other coefficients
(as well as the coefficients of the polynomial $u_k(x)$)
from the least-squares solution to the associated Sylvester linear system
of equations. Its well known superfast divide and conquer solution involves
order of $n\log^2(n)$  arithmetic operations
(cf. \cite[Chapter 5]{P01}),
but the recent numerically stable algorithm
of \cite{XXCB} accelerated by a factor of $\log (n)$
in \cite[Section 9.8]{P15} involves only $O(n\log (n))$ flops.
\end{remark}


\section{Numerical Tests}\label{snmt}


Extensive numerical tests of the algorithms of this paper, performed in the Graduate
Center of the City University of New York. 
They
are the contribution of the second author 
(at some points he was assisted by Ivan Retamoso).
The tests recorded  the number of iterations and the error  
of the approximation of the real roots of 
benchmark polynomials to which we applied these algorithms.
We have recorded similar data also for the  
approximation of real eigenvalues of some random  matrices
 $M$ by means of applying Algorithms \ref{alg2} and \ref{alg2a}.  
In the latter case the convergence 
of these algorithms and the number of their iterations
depended 
mostly on the characteristic polynomials of $M$,
even though the estimates for the arithmetic cost of performing each 
iteration generally grew compared to the special case where $M=C_p$.

In some cases we stopped the iterations 
already when they produced crude 
approximation to the roots. This is
because, instead of continuing the iterations, we can apply
the algorithms of \cite{PT14}
followed by 
Newton's or Ehrlich--Aberth's iterations
(cf. Section \ref{sfitr}), which
refine very fast 
these crude approximations.

Finally we note that the test
results  in the present section are quite encouraging
(in spite of our caveat in Remark \ref{renrei}),
e.g., the numbers of iterations 
required for convergence of our algorithms
have grown very slowly (if at all) when we increased 
the degree of the input polynomials and 
dimension of the input  matrices from 64 to 1024.
 We performed all tests with the IEEE standard double precision.

The implementation is available upon request.


\subsection{Tests for the Modified Matrix Sign Iterations (Algorithm  \ref{alg2})}\label{snmtms}


In the first series of the tests, Algorithm \ref{alg2} has been applied to one of the Mignotte
benchmark polynomials, namely $p(x)=x^n + (100x-1)^3$. It is known that this
polynomial has three ill-conditioned roots clustered about $0.01$
and has $n-3$ well-conditioned  roots. In the tests, 
Algorithm \ref{alg2} has output
the roots within the error less than $10^{-6}$ 
by using 9 iterations, for $n=32$ and  $n=64$
and by using 11 iterations, for $n=128$ and $n=256$.

In the second series of the tests, polynomials $p(x)$ of degree
$n=50,100,150,200$, and $250$ have been generated as the products
$p(x)=p_{1}(x)p_{2}(x)$, for the $r$th degree  Chebyshev polynomial  $p_{1}(x)$
(having $r$ real roots),  $r=8,12,16$, 
and $p_{2}(x)=\sum^{n-r}_{i=0}a_{i}x^{i}$, $a_{j}$ being i.i.d. standard Gaussian random variables,
for  $j=0,\dots,n-r$.
Algorithm  \ref{alg2} (performed with  double precision) was
applied to 100 such polynomials $p(x)$, for each pair
of $n$ and $r$.
 Table \ref{tabhank} displays the output data,
 namely, the average values and the standard deviation of 
the numbers of iterations and of the maximum difference between the output values
of the roots and their values produced by MATLAB root-finding
function "roots()".

In the third series of the tests,
Algorithm \ref{alg2} 
 approximated the real eigenvalues 
$x_{1},\dots,x_{r}$
of a  random complex symmetric matrix
 $A=U^{T}\Sigma U$, for 
$\Sigma=\diag(x_{1},\dots,x_{r},y_{1},\dots, y_{n-r})$,  
 $r$ i.i.d.  real standard Gaussian random variables
 $x_{1},\dots,x_{r}$, 
 $n-r$ i.i.d. complex (non-real) 
standard  Gaussian random variables
 $y_{1},\dots, y_{n-r}$, and 
 a $n\times n$ standard 
Gaussian random orthogonal matrix $U$.
 Table \ref{tabhank2} displays the mean and standard deviation of
the number of iterations and the error bounds in these tests,
for $n=50,100,150,200,250$ and $r=8,12,16$.
  
In order to estimate the number of iterations required in our 
algorithms, we periodically estimated the numerical
rank of the associated matrix  in every $k$ successive iterations,
for $k=5$ in most of our experiments.

\begin{table}[h]
\caption{Number of Iterations and Error Bounds for 
Algorithm \ref{alg2} on Random Polynomials}
\label{tabhank}
  \begin{center}

    \begin{tabular}{| c | c | c | c | c | c | }

\hline

\bf{n} & \bf{r} & \bf{Iteration-mean} & \bf{Iteration-std} & \bf{Error-mean} & \bf{Error-std} \\ \hline

$50$ & $8$ & $7.44$ & $1.12$ & $4.18\times 10^{-6}$ & $1.11\times 10^{-5}$\\ \hline
$100$ & $8$ & $8.76$ & $1.30$ & $5.90\times 10^{-6}$ & $1.47\times 10^{-5}$\\ \hline
$150$ & $8$ & $9.12$ & $0.88$ & $2.61\times 10^{-5}$ & $1.03\times 10^{-4}$\\ \hline
$200$ & $8$ & $9.64$ & $0.86$ & $1.48\times 10^{-6}$ & $5.93\times 10^{-6}$\\ \hline
$250$ & $8$ & $9.96$ & $0.73$ & $1.09\times 10^{-7}$ & $5.23\times 10^{-5}$\\ \hline

$50$ & $12$ & $7.16$ & $0.85$ & $3.45\times 10^{-4}$ & $9.20\times 10^{-4}$\\ \hline
$100$ & $12$ & $8.64$ & $1.15$ & $1.34\times 10^{-5}$ & $2.67\times 10^{-5}$\\ \hline
$150$ & $12$ & $9.12$ & $2.39$ & $3.38\times 10^{-4}$ & $1.08\times 10^{-3}$\\ \hline
$200$ & $12$ & $9.76$ & $2.52$ & $6.89\times 10^{-6}$ & $1.75\times 10^{-5}$\\ \hline
$250$ & $12$ & $10.04$ & $1.17$ & $1.89\times 10^{-5}$ & $4.04\times 10^{-5}$\\ \hline

$50$ & $16$ & $7.28$ & $5.06$ & $3.67\times 10^{-3}$ & $7.62\times 10^{-3}$\\ \hline
$100$ & $16$ & $10.20$ & $5.82$ & $1.44\times 10^{-3}$ & $4.51\times 10^{-3}$\\ \hline
$150$ & $16$ & $15.24$ & $6.33$ & $1.25\times 10^{-3}$ & $4.90\times 10^{-3}$\\ \hline
$200$ & $16$ & $13.36$ & $5.38$ & $1.07\times 10^{-3}$ & $4.72\times 10^{-3}$\\ \hline
$250$ & $16$ & $13.46$ & $6.23$ & $1.16\times 10^{-4}$ & $2.45\times 10^{-4}$\\ \hline

\end{tabular}

  \end{center}

\end{table}

\begin{table}[h]
\caption{Number of Iterations and Error Bounds for Algorithm 
 \ref{alg2} on Random Matrices}
\label{tabhank2}
  \begin{center}

    \begin{tabular}{| c | c | c | c | c | c | }

\hline

\bf{n} & \bf{r} & \bf{Iteration-mean} & \bf{Iteration-std} & \bf{Error-mean} & \bf{Error-std} \\ \hline

$50$ & $8$ & $10.02$ & $1.83$ & $5.51\times 10^{-11}$ & $1.65\times 10^{-10}$\\ \hline
$100$ & $8$ & $10.81$ & $2.04$ & $1.71\times 10^{-12}$ & $5.24\times 10^{-12}$\\ \hline
$150$ & $8$ & $14.02$ & $2.45$ & $1.31\times 10^{-13}$ & $3.96\times 10^{-13}$\\ \hline
$200$ & $8$ & $12.07$ & $0.94$ & $2.12\times 10^{-11}$ & $6.70\times 10^{-11}$\\ \hline
$250$ & $8$ & $13.59$ & $1.27$ & $2.75\times 10^{-10}$ & $8.14\times 10^{-10}$\\ \hline

$50$ & $12$ & $10.46$ & $1.26$ & $1.02\times 10^{-12}$ & $2.61\times 10^{-12}$\\ \hline
$100$ & $12$ & $10.60$ & $1.51$ & $1.79\times 10^{-10}$ & $3.66\times 10^{-10}$\\ \hline
$150$ & $12$ & $11.25$ & $1.32$ & $5.69\times 10^{-8}$ & $1.80\times 10^{-7}$\\ \hline
$200$ & $12$ & $12.36$ & $1.89$ & $7.91\times 10^{-10}$ & $2.50\times 10^{-9}$\\ \hline
$250$ & $12$ & $11.72$ & $1.49$ & $2.53\times 10^{-12}$ & $3.84\times 10^{-12}$\\ \hline

$50$ & $16$ & $10.10$ & $1.45$ & $1.86\times 10^{-9}$ & $5.77\times 10^{-9}$\\ \hline
$100$ & $16$ & $11.39$ & $1.70$ & $1.37\times 10^{-10}$ & $2.39\times 10^{-10}$\\ \hline
$150$ & $16$ & $11.62$ & $1.78$ & $1.49\times 10^{-11}$ & $4.580\times 10^{-11}$\\ \hline
$200$ & $16$ & $11.88$ & $1.32$ & $1.04\times 10^{-12}$ & $2.09\times 10^{-12}$\\ \hline
$250$ & $16$ & $12.54$ & $1.51$ & $3.41\times 10^{-11}$ & $1.08\times 10^{-10}$\\ \hline

\end{tabular}

  \end{center}

\end{table}

\clearpage

\subsection{Tests for the Stabilized Matrix Sign Iterations (Algorithm \ref{alg2a}) \\
Applied to Polynomials}

We tested Algorithm  \ref{alg2a} on various modified benchmark polynomials from the website of MPSolve (http://numpi.dm.unipi.it/mpsolve-2.2/). 
With the exception of the polynomials of Type IV below, 
we tested benchmark polynomials that had only trivial real roots 0 and $\pm 1$,
and  we
multiplied them by Chebyshev polynomials
of degree $r$, for $r=8,12$, and $16$, which have only
 real roots.

Having generated such a polynomial $p=p(x)$ 
and its companion matrix $C_{p}$, 
we computed the condition numbers of the matrices
$N_{k}=C_{p}+2^{7+k}I_{n}$ with $k=1,2,\dots$ and selected 
 an 
integer $k$ such that $\kappa(N_{k})<10^{5}$. 
Clearly, this is ensured for sufficiently 
large integers $k$ defining 
  diagonally dominant matrices $N_{k}$,
but in our tests $k$ was less than five in most cases. 

Having fixed $k$ and $N_k$ and following 
the description of Algorithm \ref{alg2a}, we computed at first the matrices
$Y_{1}=\alpha I_{n}+N_{k}$ and $Y_{2}=\alpha I_{n}-N_{k}$, for $\alpha=0.0001 \sqrt {-1}$, and 
then successively the matrices
 $Y_{i+1,j}=\frac{1}{2}(Y_{i,j}-Y_{i,j}^{-1})$ with $Y_{0,j}=Y_{j}$, for $j=1,2$
(cf. Algorithms  \ref{alg2} and  \ref{alg2a}). 
 
We have observed  that with our real shifts
by $2^{7+k}I_{n}$ at the initial stage, 
non-real eigenvalues of $Y_{1}$ and $Y_{2}$ were never close to $\pm \sqrt {-1}$ 
at the first $7+k$ iterations. 
So we began checking convergence
only when we
have performed these initial iterations, and 
since that moment 
we checked convergence in every five iterations. 
As soon as we observed that 
$\nrank(Y'_{i})=r$, for
$Y'_{i}=Y_{i,1}+Y_{i,2}$ and for
$r$ denoting the number of distinct real roots of $p(x)$,
$r=8,12$,$16$, 
 we stopped  the iterations and moved to
the final stage of the algorithm, that is,
approximated the real eigenvalues of matrix $C_{p}$, equal to the real roots of
the polynomial $p(x)$.

We have run numerical tests on polynomials of  five types having
degree $n=64,128,256,512,1024$, and we compared our results with the outputs 
of MATLAB function "roots()". 
We defined polynomial $p(x)$ of Types I--III and V
as the products $p(x)=p_{1}(x)p_{2}(x)$ where $p_{1}(x)$ is the $r$-th  degree Chebyshev polynomial and $p_2(x)$ are the following polynomials:

I. $p_{2}(x)=x^{n-r}-1$,

II. $p_{2}(x)=1+2x+3x^{2}+\cdots+(n-r+1)x^{n-r}$,

III. $p_{2}(x)=(x+1)(x+a)(x+a^{2})\cdots(x+a^{n-r-1})$, with $a=\frac{i}{100}$,
and

V. $p_{2}(x)=\sum^{n}_{k=0}a_{k}x^{k}$, with $a_{0},\dots,a_{n}$ being i.i.d. standard random variables. \\
We also tested the following polynomials of Type IV,

IV. $p(x)=x^{n-r}-(ax-1)^{3}$, where $a=60,80,100$.

 Tables \ref{test1}--\ref{test5} display 
the number of iterations and the maximum error bounds, 
for the polynomials of Types I--IV
(cf. our Remark \ref{renrei}). 
Table \ref{test2a} shows 
the average error bounds and the average numbers of iterations
in 50 tests with the polynomials of Type V.

\begin{table}[!hbp]
\caption{Number of Iterations and Error Bounds for Algorithm \ref{alg2a} on Type I Polynomials}
\label{test1}
  \begin{center}

    \begin{tabular}{| c | c | c | c | c | c | }
\hline
\bf{n} & \bf{r} & \bf{Iterations} & \bf{Errors} \\ \hline

$64$&		$8$&		$10$&$	1.03E-10$ \\ \hline
$64$&		$12$&		$23$&$	1.32E-08$ \\ \hline
$64$&		$16$&		$23$&$	3.97E-06$ \\ \hline	
$128$&	$8$	&	$10$&$	1.60E-10$ \\ \hline
$128$&	$12$&		$23$&$	4.91E-04$ \\ \hline
$128$&	$16$&		$23$&$	2.22E-03$ \\ \hline
$256$&	$8$&		$10$&$	6.18E-06$ \\ \hline
$256$&	$12$&		$28$&$	1.75E-09$ \\ \hline
$256$&	$16$&		$28$&$	3.54E-06$ \\ \hline
$512$&	$8$&		$15$&$	8.05E-13$ \\ \hline
$512$&	$12$&		$28$&$	1.71E-08$ \\ \hline
$512$&	$16$&		$28$&$	2.78E-05$ \\ \hline
$1024$&	$8$&		$15$&$	2.33E-12$ \\ \hline
$1024	$&	$12$&		$28$&$	1.27E-09$ \\ \hline
$1024	$&	$16$&		$28$&$	2.19E-05$ \\ \hline

\end{tabular}

  \end{center}

\end{table}

\begin{table}[!hbp]
\caption{Number of Iterations and Error Bounds for Algorithm \ref{alg2a} on Type II Polynomials}
\label{test2}
  \begin{center}

    \begin{tabular}{| c | c | c | c | c | c | }
\hline
\bf{n} & \bf{r} & \bf{Iterations} & \bf{Errors} \\ \hline

$64$&		$8$&		$10$&$	1.53E-11$ \\ \hline
$64$&		$12$&		$23$&$	1.30E-07$ \\ \hline
$64$&		$16$&		$23$&$	1.40E-05$ \\ \hline	
$128$&	$8$	&	$28$&$	9.42E-11$ \\ \hline
$128$&	$12$&		$10$&$	7.51E-08$ \\ \hline
$128$&	$16$&		$28$&$	2.27E-04$ \\ \hline
$256$&	$8$&		$28$&$	1.92E-11$ \\ \hline
$256$&	$12$&		$28$&$	2.21E-07$ \\ \hline
$256$&	$16$&		$28$&$	1.69E-03$ \\ \hline
$512$&	$8$&		$28$&$	3.68E-12$ \\ \hline
$512$&	$12$&		$28$&$	2.17E-06$ \\ \hline
$512$&	$16$&		$33$&$	1.53E-02$ \\ \hline
$1024$&	$8$&		$28$&$	2.96E-11$ \\ \hline
$1024	$&	$12$&		$33$&$	5.00E-07$ \\ \hline
$1024	$&	$16$&		$33$&$	3.58E-03$ \\ \hline

\end{tabular}

  \end{center}

\end{table}

\begin{table}[!hbp]
\caption{Number of Iterations and Error Bounds for Algorithm 
\ref{alg2a} on Type III Polynomials}
\label{test3}
  \begin{center}

    \begin{tabular}{| c | c | c | c | c | c | }
\hline
\bf{n} & \bf{r} & \bf{Iterations} & \bf{Errors} \\ \hline

$64$&		$8$&		$28$&$	4.63E-11$ \\ \hline
$64$&		$12$&		$23$&$	1.69E-07$ \\ \hline
$64$&		$16$&		$28$&$	7.36E-06$ \\ \hline	
$128$&	$8$	&	$28$&$	3.83E-12$ \\ \hline
$128$&	$12$&		$23$&$	1.45E-08$ \\ \hline
$128$&	$16$&		$28$&$	1.68E-05$ \\ \hline
$256$&	$8$&		$28$&$	1.58E-12$ \\ \hline
$256$&	$12$&		$23$&$	1.02E-04$ \\ \hline
$256$&	$16$&		$28$&$	6.50E-04$ \\ \hline
$512$&	$8$&		$28$&$	7.69E-13$ \\ \hline
$512$&	$12$&		$23$&$	5.00E-09$ \\ \hline
$512$&	$16$&		$28$&$	8.60E-06$ \\ \hline
$1024$&	$8$&		$28$&$	9.90E-14$ \\ \hline
$1024	$&	$12$&		$23$&$	1.45E-09$ \\ \hline
$1024	$&	$16$&		$28$&$	2.64E-05$ \\ \hline

\end{tabular}

  \end{center}

\end{table}

\begin{table}[!hbp]
\caption{Number of Iterations and Error Bounds Algorithm 
\ref{alg2a} on Type IV Polynomials}
\label{test5}
  \begin{center}

    \begin{tabular}{| c | c | c | c | c | c | }
\hline
\bf{n} & \bf{a} & \bf{Iterations} & \bf{Errors} \\ \hline

$64$&		$60$&		$41$&$	2.43E-04$ \\ \hline
$64$&		$80$&		$42$&$	7.98E-04$ \\ \hline
$64$&		$100$&	$43$&$	1.72E-05$ \\ \hline	
$128$&	$60$	&	$41$&$	1.12E-03$ \\ \hline
$128$&	$80$&		$42$&$	4.43E-04$ \\ \hline
$128$&	$100$&	$43$&$	1.31E-04$ \\ \hline
$256$&	$60$&		$41$&$	2.10E-04$ \\ \hline
$256$&	$80$&		$42$&$	1.91E-04$ \\ \hline
$256$&	$100$&	$43$&$	1.34E-04$ \\ \hline
$512$&	$60$&		$41$&$	3.37E-04$ \\ \hline
$512$&	$80$&		$42$&$	1.80E-04$ \\ \hline
$512$&	$100$&	$43$&$	8.33E-05$ \\ \hline
$1024$&	$60$&		$36$&$	1.10E-01$ \\ \hline
$1024	$&	$80$&		$42$&$	1.16E-04$ \\ \hline
$1024	$&	$100$&	$43$&$	1.76E-04$ \\ \hline

\end{tabular}

  \end{center}

\end{table}

\begin{table}[!hbp]
\caption{Number of Iterations and Error Bounds for Algorithm 
\ref{alg2a} on Type V Polynomials}
\label{test2a}
  \begin{center}

    \begin{tabular}{| c | c | c | c | c | c | }
\hline
\bf{n} & \bf{r} & \bf{Iterations} & \bf{Errors} \\ \hline

$128$&	$8$&		$22.3$&$	5.33E-06$ \\ \hline
$128$&	$12$&		$24.6$&$	4.85E-05$ \\ \hline
$128$&	$16$&		$24.94$&$	3.59E-03$ \\ \hline	
$256$&	$8$	&	$26.02$&$	1.11E-06$ \\ \hline
$256$&	$12$&		$27.01$&$	2.37E-05$ \\ \hline
$256$&	$16$&		$30.18$&$	1.80E-03$ \\ \hline
$512$&	$8$&		$27.54$&$	2.73E-08$ \\ \hline
$512$&	$12$&		$28.00$&$	2.27E-06$ \\ \hline
$512$&	$16$&		$38.18$&$	2.39E-03$ \\ \hline

\end{tabular}

  \end{center}

\end{table}

\clearpage

\subsection{Tests for the Stabilized Matrix Sign Iterations (Algorithm \ref{alg2a}) on Gaussian Random Matrices}

We tested Algorithm \ref{alg2a} on randomly generated matrices of two types:

Type I: Gaussian random tridiagonal matrices of dimension $n=64, 128, 256, 512, 1024$. 
We generated each entry in the tridiagonal part independently  by using standard Gaussian distribution and set the other entries to 0. Our tables show the error bounds equal to the maximal difference of the outputs of our algorithm and MATLAB function "eig()". 
We generated 100 matrices, for each $n$, and  recorded the mean and standard deviation of 
the error bounds and of the numbers of iterations.

Type II: Random matrices $A$ with a fixed number of real eigenvalues. At first we generated a diagonal matrix $\Sigma$ with $r$ diagonal 
entries under the standard real Gaussian distribution 
and $n-r$ diagonal entries under the standard 
complex Gaussian distribution,
 for $n = 64, 128, 256, 512, 1024$ and $r = 8, 12, 16$. Then we generated a   standard Gaussian random orthogonal matrix $Q$. Finally 
we computed the matrices $A=Q^{T}\Sigma Q$.
 We generated 100 such matrices $A$, for each pair of $n$ and $r$,
and  recorded the mean and standard deviation of the error bounds and of the numbers of iterations.

The following two tables summarize the performance data, showing  a low number of iterations required
for ensuring  
the approximation of the eigenvalues with a  reasonable precision. 

\begin{table}[!hbp]
\caption{Number of Iterations and Error Bounds for Root-finding Algorithm  \ref{alg2a} on Type I matrices}
\label{test10}
  \begin{center}

    \begin{tabular}{| c | c | c | c | c | }
\hline
\bf{n} & \bf{Iteration-mean} & \bf{Iteration-std} & \bf{Error-mean} & \bf{Error-std} \\ \hline

$64$&		$10.70$&$2.36$& 	$1.78E-06$& $1.14E-05$ \\ \hline	
$128$&	$12.16$&$3.34$&	$5.68E-07$& $4.49E-06$ \\ \hline
$256$&	$12,97$&$3.97$&	$3.26E-06$& $1.35E-05$ \\ \hline
$512$&	$15.46$&$9.82$&	$8.80E-04$& $8.44E-03$ \\ \hline
$1024$&	$16.52$&$10.26$&	$2.43E-03$& $2.25E-02$ \\ \hline

\end{tabular}

  \end{center}

\end{table}

\begin{table}[!hbp]
\caption{Number of Iterations and Error Bounds for  Algorithm \ref{alg2a} on Type II matrices}
\label{test11}
  \begin{center}

    \begin{tabular}{| c | c | c | c | c | c |  c |}
\hline
\bf{n} & \bf{r} & \bf{Iteration-mean} & \bf{Iteration-std} & \bf{Error-mean} & \bf{Error-std} \\ \hline

$64$&		$8$&		$11.65$&$2.47$&$	3.69E-08$&$2.29E-07$ \\ \hline
$64$&		$12$&		$11.75$&$2.50$&$	3.98E-10$&$2.71E-09$ \\ \hline
$64$&		$16$&		$11.60$&$2.45$&$	4.10E-09$&$3.88E-08$ \\ \hline	
$128$&	$8$&	           $13.75$&$2.79$&$ 1.17E-08$&$7.56E-08$ \\ \hline
$128$&	$12$&		$13.70$&$2.90$&$	4.41E-09$&$2.73E-08$ \\ \hline
$128$&	$16$&		$13.65$&$2.55$&$	1.23E-07$&$1.34E-06$ \\ \hline
$256$&	$8$&		$14.55$&$3.26$&$	5.59E-09$&$5.58E-08$ \\ \hline
$256$&	$12$&		$14.15$&$3.70$&$	1.38E-07$&$1.38E-06$ \\ \hline
$256$&	$16$&		$14.70$&$2.54$&$	3.06E-11$&$1.93E-10$ \\ \hline
$512$&	$8$&		$13.65$&$5.59$&$	5.08E-10$&$4.88E-09$ \\ \hline
$512$&	$12$&		$15.65$&$9.47$&$	7.46E-04$&$7.46E-03$ \\ \hline
$512$&	$16$&		$16.55$&$10.26$&$2.78E-03$&$5.47E-03$ \\ \hline
$1024$&	$8$&		$18.20$&$15.35$&$2.33E-10$&$1.22E-09$ \\ \hline
$1024	$&	$12$&		$20.85$&$17.60$&$1.27E-07$&$3.36E-07$ \\ \hline
$1024	$&	$16$&		$24.35$&$19.56$&$2.19E-03$&$4.33E-03$ \\ \hline

\end{tabular}

  \end{center}

\end{table}

\subsection{Tests for the Hybrid Matrix Algorithm (Algorithm \ref{alghbr}) on Benchmark Polynomials}\label{ststshbr}

We performed numerical tests of a hybrid algorithm. We began with 
 Algorithm \ref{alg2} and after sufficiently many iterations continued with
its variation avoiding matrix inversion. 

Namely, we first applied a real shift $\beta I$ to the companion matrix $C_{p}$, such that 
the matrix $M = C_{p}+\beta I$ had condition number less than $10^{5}$. 
Based on our previous tests,  we expected
that, for such inputs, at least $T = \log_{2}\beta$ iterations $M_{i+1} = \frac{1}{2}(M_{i} - M_{i}^{-1})$  
would be required in order to move
 the complex nonreal eigenvalues close enough to $\pm\sqrt{-1}$. After the first $T$ iterations, we periodically (in every 5 iterations) 
applied two iterations $M_{i+1} = \frac{1}{2}(M_{i}^{3} + 3M_{i})$, 
which converged with  cubic rate provided that all complex eigenvalues have distance less than $\frac{1}{2}$ from $\sqrt{-1}$ or $-\sqrt{-1}$. Before switching to the  iterations of 
the second type,
we performed the following transformation in order to avoid problems of numerical stability:

Step 1: Compute $P = \frac{0.5M + \sqrt{-1}~I}{0.5M + \sqrt{-1}~I}$, which maps the real line into the unit circle.

Step 2: Compute $Y = \frac{2\sqrt{-1}}{3}(P - P^{-1})$, mapping the unit circle onto the interval $[-2/3, 2/3]$.

Note that these two maps together keep the values $\pm\sqrt{-1}$ unmoved.

We tested polynomials of Types II and  IV of the previous section.  
For polynomials of Types I, III, and V, the test results
 were similar to those for 
polynomials of Type II, apparently
due to the shared Chebyshev factors. The test results on Type IV polynomials indicate the strength of this algorithm in the case of clustered roots.

%

The number of iterations required and the error bound are displayed in the tables below.

\begin{table}[!hbp]
\caption{Number of Iterations and Error Bounds for Hybrid Algorithm on Type II Polynomials}
\label{test12}
  \begin{center}

    \begin{tabular}{| c | c | c | c | c | c | }
\hline
\bf{n} & \bf{r} & \bf{Iterations} & \bf{Errors} \\ \hline

$64$&		$8$&		$10$&$	3.69E-10$ \\ \hline
$64$&		$12$&		$23$&$	4.96E-08$ \\ \hline
$64$&		$16$&		$23$&$	4.97E-03$ \\ \hline	
$128$&	$8$	&	$10$&$	2.28E-11$ \\ \hline
$128$&	$12$&		$28$&$	1.97E-07$ \\ \hline
$128$&	$16$&		$23$&$	8.68E-02$ \\ \hline
$256$&	$8$&		$28$&$	6.56E-12$ \\ \hline
$256$&	$12$&		$28$&$	3.64E-07$ \\ \hline
$256$&	$16$&		$28$&$	3.82E-04$ \\ \hline
$512$&	$8$&		$15$&$	8.05E-13$ \\ \hline
$512$&	$12$&		$28$&$	1.71E-08$ \\ \hline
$512$&	$16$&		$28$&$	2.78E-05$ \\ \hline
$1024$&	$8$&		$28$&$	3.72E-11$ \\ \hline
$1024	$&	$12$&		$28$&$	1.09E-08$ \\ \hline
$1024	$&	$16$&		$33$&$	2.19E-05$ \\ \hline

\end{tabular}

  \end{center}

\end{table}

\begin{table}[!hbp]
\caption{Number of Iterations and Error Bounds for Hybrid Algorithm on Type IV Polynomials}
\label{test13}
  \begin{center}

    \begin{tabular}{| c | c | c | c | c | c | }
\hline
\bf{n} & \bf{Iterations} & \bf{Errors} \\ \hline

$64$&		$33$&$	7.32E-05$ \\ \hline
$128$&	$33$&$	6.12E-06$ \\ \hline
$256$&	$38$&$	1.60E-05$ \\ \hline
$512$&	$38$&$	1.08E-04$ \\ \hline
$1024$&	$38$&$	9.19E-01$ \\ \hline

\end{tabular}

  \end{center}

\end{table}

\subsection{Tests for the Modular Square Root Iterations (Algorithm \ref{alg3})}

Table \ref{test16} displays our test results for  Algorithm  \ref{alg3}, that is, 
for the iterations $f_{i+1}(x) \equiv \frac{1}{2}(f_i(x) - f_i(x)^{-1})\mod p(x)$,
which computed polynomial inverses modulo $p(x)$ by solving the 
associated Sylvester linear systems of equations. 
We applied the tests  to polynomials of Types I and II.

Already after a small number of iterations, that is, for small integers $i$,
the tests have consistently
 produced polynomials
 $f_i(x)$ whose roots approximated  the complex roots of 
the polynomial $p(x)$ of (\ref{eqpoly})
within the fixed 
tolerance bound $\epsilon = 10^{-5}$.   
At this stage of our tests we 
applied the MATLAB function "roots()" 
in order to avoid actual computation of agcds.
Namely, as soon as we observed that
the polynomial $p(x)$ shared all its complex roots with
the polynomial $f_i(x)$, we stopped the iterations.

\begin{table}[!hbp]
\caption{Number of Iterations for Algorithm  \ref{alg3} on Polynomials of Types I and II}
\label{test16}
  \begin{center}

    \begin{tabular}{| c | c | c | c | c | c | }
\hline
n	&	r	&     Type I & Type II	\\ \hline
64	&	8	&	9   &	14  \\ \hline
64	&	12	&	4   &	16  \\ \hline
64	&	16	&	2   &	17  \\ \hline
128	&	8	&	9   &	14  \\ \hline
128	&	12	&	12  &	16  \\ \hline
128	&	16	&	2   &	17  \\ \hline
256	&	8	&	9   &	14  \\ \hline
256	&	12	&	12  &	16  \\ \hline
256	&	16	&	8   &	17  \\ \hline
512	&	8	&	9   &	14  \\ \hline
512	&	12	&	12  &	16  \\ \hline
512	&	16	&	8   &	17  \\ \hline
1024	&	8	&	10  & 	14  \\ \hline
1024	&	12	&	12  &	16  \\ \hline
1024	&	16	&	11  &	17  \\ \hline

\end{tabular}

  \end{center}

\end{table}



{\bf Acknowledgements:} 
  This work has been supported by NSF Grant CCF--1116736 
and PSC CUNY Award 67699-00 45. 
We are also grateful 
to Dario A. Bini and two anonymous reviewers, 
for thoughtful and  helpful comments and to
Ioannis Z. Emiris and Bernard Mourrain 
for pointing us out the bibliography on the 
distribution of real roots of a polynomial.


\bigskip



{\bf {\LARGE {Appendix}}}
\appendix 



\section{Some Maps of the Variables and the Roots}\label{smpgvr}


Some basic maps of the roots of a polynomial can be computed
at a linear or nearly linear arithmetic cost.
 
\begin{theorem}\label{thshsc} ({\em  Root Inversion, Shift and Scaling}, cf. \cite{P01}.)

(i) Given a polynomial $p(x)$ of (\ref{eqpoly}) and two scalars $a$ and $b$, one can
compute the coefficients of the polynomial $q(x)=p(ax+b)$ by using
$O(n \log (n))$  arithmetic operations. 
This bound decreases to $2n-1$ multiplications if $b=0$.  

(ii) Reversing a polynomial inverts all its roots by involving no flops,
that is, $$p_{\rm rev}(x)=x^np(1/x)=\sum_{i=0}^np_ix^{n-i}=p_n\prod_{j=1}^n(1-xx_j).$$
\end{theorem}

Note that by shifting and scaling
the variable,
 we can 
move
 all roots of $p(x)$ into a fixed disc, e.g., $D(0,1)=\{x:~|x|\le 1\}$.

\begin{theorem}\label{thdnd} ({\em Dandelin's Root Squaring}, cf. \cite{H59}.)

(i) Let a polynomial $p(x)$ of (\ref{eqpoly}) be monic.
Then $q(x)=(-1)^np(\sqrt{x}~)p(-\sqrt{x}~)=\prod_{j=1}^n(x-x_{j}^2)$.

(ii) One can
evaluate $p(x)$ at the $k$-th roots of unity, for $k>2n$, and then
interpolate to $q(x)$ by using $O(k\log (k))$ 
 arithmetic operations  overall.
\end{theorem}


\begin{remark}\label{rercrs} 
Recursive root-squaring is prone to numerical stability problems 
because the coefficients of the iterated polynomials very quickly span many orders of magnitude.  
It is somewhat  surprising, but the Boolean complexity of
the recursive root-squaring process is relatively reasonable 
if high output precision
is required
  \cite{P95},
\cite{P02}.
Moreover, one can avoid numerical  stability problems 
and perform all iterations with double precision
 by applying a special tangential representation of the coefficients and 
of the intermediate results proposed in  \cite{MZ01}. In this case 
the computations involve more general operations than flops; in terms of 
the CPU time 
the computational cost per
iteration has the same order as
 $n^2$ flops, performed with double precision. 
\end{remark}

\begin{theorem}\label{thcl} ({\em The Cayley Maps.})

(i) The map
$y=(x -a\sqrt{-1}~)/(x+a\sqrt{-1}~)$, for any real nonzero scalar $a$,
sends the real axis $\{x:~x~{\rm is~real}\}$
onto the unit circle $C(0,1)=\{y:~|y|=1\}$.

(ii) The converse map $x=a\sqrt{-1}~(1-y )/(y+1)$ sends 
the unit circle $C(0,1)$ onto the real axis.
\end{theorem}


\section{Some Functional Iterations for Polynomial Root-finding}\label{sfitr} 


Newton's and Ehrlich--Aberth's are two celebrated functional iteration
processes for the approximation 
of a single root of a
polynomial $p(x)$ of (\ref{eqpoly}) and all its roots, respectively.
They are highly efficient and popular, but 
 not specialized to our task of approximating real roots,
and we only use them as auxiliary root-refiners. 

Hereafter a disc $D(X,r)$ is said to be $\gamma$-{\em isolated} 
for a polynomial $p(x)$ and $\gamma>1$ if it contains 
all roots of the polynomial lying in the disc $D(X,\gamma r)$.
In this case we say that the disc has 
 {\em isolation ratio} at least $\gamma$.

Newton's iterations refine an approximation $y^{(0)}$
to a single root of a
polynomial $p(x)$ of (\ref{eqpoly}),
 \begin{equation}\label{eqnewt}
y_0=c,~y^{(h+1)}=y^{(h)}-p(y^{(h)})/p'(y^{(h)}),~h=0,1,\dots
\end{equation}  
 
Ehrlich--Aberth's iterations refine $n$ simultaneous approximations 
$z_1^{(0)},\dots, z_1^{(n)}$ 
to all $n$ roots $x_1,\dots,x_n$ of such a polynomial,  
\begin{equation}\label{eqehrab}
z_i^{(h+1)}=z_i^{(h)}-1/e_i^{(h)},~{\rm for}~
 e_i^{(h)}=p(z_i^{(h)})/ p'(z_i^{(h)})-\sum_{j\neq i}\frac{1}{ z_i^{(h)}- z_j^{(h)}},~i=1,\dots,n,
\end{equation}
See  \cite{M07}, \cite{MP13} for various other functional iterations.

As we can see next, both iterative algorithms
 refine  very fast the crude initial approximations
to simple isolated roots of a polynomial.


\begin{theorem}\label{thren} 
 Assume a polynomial $p=p(x)$ of  (\ref{eqpoly})
and let $0<3(n-1)|y_0-x_1|<|y_0-x_j|$, for $j=2,\dots,n$.
  Then  Newton's
  iterations (\ref{eqnewt})
  converge to the root $x_1$ quadratically
  right from the start, namely,  $|y_k-x_1|\le 2|y_0-x_1|/2^{2^k}$,
for $k=0,1,\dots$.
\end{theorem}
\begin{proof} 
See \cite[Theorem 2.4]{T98}, which strengthens 
\cite[Corollary~4.5]{R87}.
\end{proof} 


\begin{theorem}\label{thtil} (See \cite[Theorem 3.3]{T98}.)
 Assume a polynomial $p=p(x)$ of  (\ref{eqpoly})
and crude initial approximations  $y_j^{(0)}$ to the roots $x_j$ such that
 $0<3\sqrt{n-1}~|y_j^{(0)}-x_j|<|y_j^{(0)}-y_i^{(0)}|$,
 for $i\neq j$, $j=1,\dots,n$.
  Then  Ehrlich--Aberth's
  iterations 
  converge to the roots $x_j$ with the cubic rate
  right from the start, namely,  
$|y^{(k)}_j-x_j|\le |y_j^{(0)}-x_j|/(2^{3^k}\sqrt {(n-1)}~)$,
for $j=1,\dots,n$ and $k=0,1,\dots$.
\end{theorem}

The paper \cite{T98} also proves quadratic convergence of the WDK
iterations to all $n$ roots, lying in some given discs with 
an isolation ratios at least $3(n-1)/8$. 
These iterations are due to Weierstrass \cite{W03},
but are frequently attributed to its later re-discoveries by Durand in 1960 and Kerner in 1966.

By exploiting the 
correlations between the coefficients of a polynomial 
and the power sums of its roots,
the paper  \cite{PT14}  had weakened
the above assumptions on the initial isolation.
More precisely,  
assuming that  a simple root lies
 in the disc $D(0,1)$ and that the disc has an isolation ratio at least 
 $s\ge 1+1/\log_2(n)$,
the paper  \cite{PT14} increases it to $cn^d$,
for any fixed pair of constants $c$ and $d$, at the arithmetic cost 
$O(n)$, and similarly increased the
 isolation ratio of the  $n$ discs
covering   
 all the
$n$ roots 
at  the  arithmetic cost $O(n\log^2(n))$.

In the case of a single disc, one 
can assume even an isolation ratio $s\ge 1+c'/n^{d'}$,
for any pair of constants $c'$ and $d'$, and then increase it
to $s\ge cn^d$, for any 
other pair of constants $c$ and $d$, at the arithmetic cost 
$O(n\log^2 (n))$. Indeed one can achieve this by performing $h$
root-squaring iterations of Theorem \ref{thdnd}, for $h$
of order $\log (n)$
because each 
squaring of the roots also squares the isolation ratio.
 This lifting process ensures the desired  isolation for the lifted roots of the new 
lifted polynomial,
but the descending back to the original
roots can be also achieved by using $O(n\log^2(n))$ 
 arithmetic operations \cite{P95},
\cite{P02}. We refer the reader to
 Remark \ref{rercrs} on the precision growth in these
iterations and their Boolean complexity.

Can we completely relax the assumption of the initial isolation? 
Empirically fast global convergence
(that is, convergence right from the start) is  very strong
 over all inputs for
 the WDK, Ehrlich--Aberth, and some other iterations that approximate 
simultaneously all $n$ roots of a polynomial $p(x)$ of (\ref{eqpoly}).
The papers \cite{P11}, \cite{PZ11a}, and \cite{P12} have challenged 
the researchers to support 
this observation with a formal proof, which is still missing, however. 


\section{Fast Root-finding Where All Roots Are Real}\label{sallr}


\begin{theorem}\label{thmlg} 
Assume that all roots 
of a polynomial  $p(x)$ of (\ref{eqpoly}) are real. 

(i) Then the  modified Laguerre algorithm of \cite{DJLZ97} converges
to all of them right from the start,
 uses $O(n)$ flops per iteration, and therefore
approximates all the $n$ roots within $\epsilon =1/2^b$
by using $O(\log (b))$ iterations and
 performing $O(n\log (b))$ flops.

(ii) The latter asymptotic arithmetic cost bound is optimal
and is supported by the
 alternative algorithms of \cite{BT90} and \cite{BP98} as well.

(iii) All these algorithms reach the optimal Boolean cost 
bound up to polylogarithmic factors.
\end{theorem}


\section{Counting the Roots in a Disc. Root Radii, Distances 
to the Roots, and the Proximity Tests}\label{srrd}


In this subsection we estimate the distances to the roots of $p(x)$ from a complex point
and the number of the roots in an isolated disc. 

The latter task can be solved by using the following result from
 \cite[Lemma 7.1]{R87} (cf. also \cite{H74}, \cite[Theorem 14.1]{S82} and \cite{BP00}).


\begin{theorem}\label{thwnmb}  \cite[Lemma 7.1]{R87}
It is sufficient to
perform FFT at $n'=16\lceil \log_2n\rceil$
points (using $1.5n'\log (n')$ flops) and $O(n)$ additional flops and comparisons 
of real numbers with 0 in order to compute the number of roots of a polynomial $p(x)$
of (\ref{eqpoly}) in a 9-isolated disc $D(0,r)$.
\end{theorem}


\begin{remark}\label{rewnmb} 
The algorithm of \cite{R87} supporting Theorem \ref{thwnmb}
only uses the signs of the real and imaginary parts of the
$n$ output values of FFT. For some groups of the values, the pairs of
the signs stay invariant and can be 
represented by a single pair of signs. 
Can this observation
 be exploited in order to decrease the computational cost of performing
 the algorithm?
\end{remark}

\begin{corollary}\label{cownmb}
It is sufficient to
perform $O(hn\log (n))$ flops and $O(n)$  comparisons 
of real numbers with 0 in order to compute the number of roots of a polynomial $p(x)$
of (\ref{eqpoly}) in an $s$-isolated disc $D(0,1)$,
for $s=9^{1/2^h}$ and  any positive integer $h$.
\end{corollary}

\begin{proof}
Every root-squaring of Theorem \ref{thdnd}
squares all root-radii and the isolation ratios of all discs
$D(0,r)$,
for all positive $r$. Suppose $h$ repeated squaring iterations map a polynomial 
$p(x)$ into $p_h(x)$, for which  
 the disc $D(0,1)$ is 9-isolated. Then, 
by applying Theorem \ref{thwnmb}, we can 
compute the number of roots of $p_h(x)$ in this disc,
 equal to
the number of roots of $p(x)$.  
\end{proof}

In view of
Remark \ref{rercrs}, one must 
apply the slower operations of \cite{MZ01} or 
high precision computations in order to support even
a moderately long sequence of root-squaring iterations,
but in some cases it is sufficient to 
 apply Corollary \ref{cownmb},
for small positive integers $h$.
Note that $9^{1/2^h}$ is equal to 1.3160..., for  $h=2$,
to 1.1472..., for $h=3$, to 1.0710..., for $h=4$, and 
to 1.0349..., for $h=5$.

We can use the following
result if we agree to perform computations  with extended precision.

\begin{theorem}\label{thrrd} ({\rm The Root Radii Approximation}.) 

Assume
a polynomial
 $p(x)$ of (\ref{eqpoly}) and
two real scalars $c>0$ and $d$. Define the $n$ {\em root radii}
$r_j=|x_{k_j}|$, for $j=1,\dots, n$, distinct $k_1,\dots,k_n$, and
 $r_1\ge r_2\ge \cdots\ge r_n$.
Then,  by using
$O(n \log^2 (n))$ arithmetic operations, one can compute $n$ approximations $\tilde r_j$ 
to the root radii $r_j$ such that $\tilde
r_j\le r_j\le (1+c/n^d)\tilde r_j$, for $j=1,\dots, n$.
\end{theorem}
\begin{proof} (Cf. \cite{S82},  \cite[Section 4]{P00}.) 
At first fix a sufficiently large integer $k$
and  apply $k$ times the root-squaring of Theorem \ref{thdnd},
which involves $O(kn\log(n))$  arithmetic operations. Then 
apply the algorithm of \cite{S82} 
(which uses $O(n)$ 
 arithmetic operations) in order
to approximate  within a factor
of $2n$ all root radii $r_j^{(k)}=r_j^{2^k}$,
$j=1,\dots,n$,
of the output polynomial $p_k(x)$. By taking the $2^k$-th roots, 
approximate the root radii $r_1,\dots,r_n$ 
within a factor of $(2n)^{1/2^k}$, which is $1+c/n^d$, for  $k$
of order $\log(n)$.
\end{proof}

Alternatively we can approximate the root radii 
by applying the semi-heuristic method of \cite{B96}, 
used in the packages 
 MPSolve 2000 and 2012 (cf. \cite{BF00} and \cite{BR14})
or by recursively applying Theorem \ref{thwnmb}, although
neither of these techniques support competitive complexity estimates.

The following two theorems define the largest root radius  $r_1$ of the polynomial $p(x)$.


\begin{theorem}\label{thextrrrd} (See \cite{VdS70}.)
Assume a polynomial $p(x)$ of  (\ref{eqpoly}). Write
$r_1=\max_{j=1}^{n}|x_j|$, $r_n=\min_{j=1}^{n}|x_j|$, and
$\gamma^+=\max_{i=1}^{n}|p_{n-i}/p_n|$.
Then $\gamma^+/n\le r_1\le 2\gamma^+$.
\end{theorem}


\begin{theorem}\label{threfext} (See \cite{P01a}.) 
For  $\epsilon=1/2^b>0$,
one only needs $a(n,\epsilon)=O(n+b\log (b))$ flops to compute an
approximation $r_{1,\epsilon}$ to the largest root  radius $r_1$ of $p(x)$
such that $r_{1,\epsilon}\le r_1\le 5(1+\epsilon)r_{1,\epsilon}$.
In particular, $a(n,\epsilon)=O(n)$, for $b=O(n/\log (n))$, 
and $a(n,\epsilon)=O(n\log (n))$, for $b=O(n)$.
\end{theorem} 

Both theorems can be immediately extended to
the approximation of the
 smallest root radius $r_n$ because it is 
the reciprocal of
the largest root radius of 
the reverse polynomial $p_{\rm rev}(x)=x^np(1/x)$
(cf. Theorem \ref{thshsc}).
 Moreover,  by shifting  a complex point $c$ into the origin 
(cf. Theorem  \ref{thshsc}), we can turn our estimates for
the root radii into the estimates for the {\em distances to the roots}
 from the point $c$.  
Approximation of the smallest distance from a  complex point $c$ 
to a root of $p(x)$ is called the
{\em proximity test} at the point. One can perform such a
test by applying Theorems \ref{thwnmb},
\ref{thextrrrd}, or
\ref{threfext}. 

Alternatively, for
 heuristic proximity tests {\em by action}
   at 
a  point $c$ 
or  at $n$ points,
one can
 apply Newton's iterations (\ref{eqnewt})
or an appropriate functional iterations, such as
 the Ehrlich--Aberth iterations
(\ref{eqehrab}),
and estimate the distance to the roots 
by observing  convergence or divergence of the  iterations. 

Theorem \ref{threfext} and all these iterations, including
 Newton's, Ehrlich--Aberth's and WDK's,
 can be applied 
even where a polynomial $p(x)$ is defined by a black box subroutine for its evaluation
rather than by its coefficients.


\end{document}